\documentclass[12pt]{article}
\usepackage{amsfonts}
\usepackage{amssymb,amsmath,amsthm,latexsym}
\usepackage[numbers,sort&compress]{natbib}
\usepackage{amsmath}
\usepackage{indentfirst}
\usepackage{adjustbox}

\usepackage[figuresright]{rotating}

\allowdisplaybreaks[4]
\newtheorem{thm}{Theorem}[section]
\newtheorem{prop}[thm]{Proposition}
\newtheorem{lem}[thm]{Lemma}

\newtheorem{exam}{Example}
\newtheorem{defi}[thm]{Definition}
\newtheorem{rem}[thm]{Remark}

\numberwithin{equation}{section}

\begin{document}

\title{Linear complementary pairs of codes over a finite non-commutative Frobenius ring}
\author{Sanjit Bhowmick$^1$, Xiusheng Liu$^{2}$\footnote{Corresponding author.}
\setcounter{footnote}{-1}
\footnote{E-Mail addresses:
 sanjitbhowmick@niser.ac.in  (S. Bhowmick),
 lxs6682@163.com (X. Liu)}}
\date{\small
1. School of Mathematical Sciences, National Institute of Science Education and Research,\\
An OCC of Homi Bhabha National Institute,
Bhubaneswar, Odisha 752050, India\\
2. School of Science and Technology,\\ College of Arts and Science of Hubei Normal University, Huangshi, Hubei 435109, China}
\maketitle

\begin{abstract}
In this paper, we study linear complementary pairs (LCP) of codes over finite non-commutative local rings. We further provide a necessary and sufficient condition for a pair of codes $(C,D)$ to be LCP of codes over finite non-commutative Frobenius rings. The minimum distances $d(C)$ and $d(D^\perp)$ are defined as the security parameter for an LCP of codes $(C, D).$ It was recently demonstrated that if $C$ and $D$ are both $2$-sided LCP of group codes over a finite commutative Frobenius rings, $D^\perp$ and $C$ are permutation equivalent in \cite{LL23}. As a result, the security parameter for a $2$-sided group LCP $(C, D)$ of codes is simply $d(C)$. Towards this, we deliver an elementary proof of the fact that for a linear complementary pair of codes $(C,D)$, where $C$ and $D$ are linear codes over finite non-commutative Frobenius rings, under certain conditions, the dual code $D^\perp$ is equivalent to $C.$
\end{abstract}

\bf Key Words\rm : Finite non-commutative Frobenius ring, Complement submodule, Essential submodule, Injective hull.

\bf Mathematics Subject Classification\rm : 51E22; 94B05.

\section{Introduction}
Linear complementary pairs (LCP) of codes are extensively explored because of their unique algebraic structure and wide application in cryptography. This concept was first introduced by Bhasin et al., in \cite{Bhasin2015}.
LCP of codes over a finite field is further studied in \cite{Bringer2014} and \cite{Carlet2018}. Hu et al. developed these codes over finite chain rings in \cite{Liu2021}. Recently, Liu et al., \cite{LL23} extended the same results over finite commutative Frobenius rings. In particular, if $D=C^\perp$ over finite commutative Frobenius rings, LCP is reduced to LCD, studied in \cite{Bhowmick}. This paper will develop an algebraic structure of LCP of codes over finite non-commutative Frobenius rings. \\
On the other hand, Carlet et al., \cite{Carlet2018} showed that if the pair $(C,D)$ is LCP, where $C$ and $D$ are both cyclic codes over a finite field, then $C$ and $D^\perp$ are equivalent. They further showed that if the length of the codes is relatively prime to the characteristic of the finite field and $C$ and $D$ are $2D$ cyclic codes, then $C$ and $D^\perp$ are equivalent. Later, Guneri et al., \cite{Gu2018} extended the same results for linear codes $C$ and $D$, which are $mD$ cyclic codes, where $m\in\mathbb{N}.$\\
Without any restriction on the order of the group with a characteristic of a finite field, in \cite{Bor1}, Borello et al. established that if $C$ and $D$ both are group codes and the pair $(C,D)$ is an LCP of group codes, then $C$ is permutation equivalent to $D^\perp$. Further, Guneri et al., \cite{cem20} extended those results for LCP of group codes over finite commutative chain rings. Using the same method of \cite{cem20}, Liu et al. established the same results for LCP of group codes over finite commutative Frobenius rings in \cite{LL23}. They also established that $C$ and $D^\perp$ are equivalent if the pair $(C,D)$ are an LCP of group codes. In this paper, we will show that  $C$ and $D^\perp$ are equivalent if the pair $(C,D)$ are an LCP of codes over a finite non-commutative Frobenius ring under a certain condition. The same technique will be held for LCP of codes over finite fields and an LCP of codes over finite commutative Frobenius rings\\
The paper is organized as follows. In Section 2, we recall background materials on non-commutative Frobenius rings and linear codes over a finite non-commutative ring. In Section 3, we study LCP codes over finite non-commutative local rings. Furthermore, we develop LCP codes over a finite non-commutative Frobenius ring in Section $4.$ Finally, we obtain if $C$ and $D$ are two linear codes over finite non-commutative Frobenius rings, then $C$ is equivalent to $D^\perp$ under certain conditions in Section $5.$

\section{Some preliminaries}
    In this section, we will state some basic definitions and results that are needed to derive our main results. Throughout this paper, we will assume all rings to be non-commutative unless mentioned otherwise. Now let  $R$ denote a finite ring with unity $1_R.$   A right $R$-module   $A$ is said to be free if it is isomorphic to the right $R$-module $R^{t}$ for some positive integer $t.$ A right $R$-module $P$ is said to be projective if there exists another right $R$-module $Q$ such that $P\oplus Q$ is a free right $R$-module. Further, a right $R$-module $I$ is said to be injective if  for any monomorphism $g:A\rightarrow B$ of right $R$-modules and any $R$-module homomorphism $h:A\rightarrow I$, there exists an $R$-module homomorphism $h':B\rightarrow I$ such that $h=h'\circ g$. Now, the following results are well-known.
\begin{prop}\cite[Prop. 36]{Dum04}\label{th-0.0001}
A right $R$-module $I$ is injective if and only if for any right ideal, $P$ of $R$, any $R$-module homomorphism $f:P\rightarrow I$ can be extended to the $R$-module homomorphism $f':R\rightarrow I$.
\end{prop}
 \begin{prop}\cite[Prop. 2.5 and 3.4]{Lam19}\label{th-1.001} Let $P=\bigoplus\limits_{i=1}^{n} P_i$ be a direct sum of right $R$-modules $P_1,P_2,\cdots,P_n.$ The following hold.
\begin{itemize}
\item[(a)] The right $R$-module $P$ is projective if and only if each $P_i$ is a projective right $R$-module for $1 \leq i \leq n.$
\item[(b)] The right $R$-module $P$ is injective if and only if each $P_i$ is an injective right $R$-module for $1 \leq i \leq n.$
\end{itemize}
\end{prop}

The ring $R$ is said to be right (\textit{resp.} left)  self-injective if it is injective over itself as a right (\textit{resp.} left) $R$-module. Further, the finite ring $R$ is said to be Frobenius if it is right and left self-injective. Equivalently, the finite ring $R$ is said to be Frobenius if the right $R$-module $\dfrac{R}{J(R)}$ is isomorphic to $\textit{Soc}(R)$ as a right $R$-module, where $\texttt{J}(R)$ denotes the Jacobson radical of $R$ (\textit{i.e.,} the intersection of all maximal ideals of $R$)  and  $Soc(R)$ denotes the Socle of $R$ (\textit{i.e.,}  the sum of all irreducible ideals of $R$). For more details, see  \cite{Lam19}. Note that $\texttt{J}(R)$ is a $2$-sided ideal of $R.$

\begin{prop}\cite[Th. 15.9]{Lam19}\label{th-0.02} The following two statements are equivalent:
\begin{itemize}
\item[(a)] The ring $R$ is Frobenius.
\item[(b)] Every right ideal of $R$ is a projective right $R$-module if and only if it is an injective right $R$-module.
\end{itemize}
\end{prop}

A non-empty subset $C$ of $R^n$ is called a linear code of length $n$ over $R$ if it is a right $R$-submodule of $R^n$.
We define the inner product on $R^n,$ as follows $$[x,y]=\sum_{i=1}^nx_iy_i,$$ where $x=(x_1,x_2,\dots,x_n)$ and $y=(y_1,y_2,\dots,y_n)$ in $R^n.$ For a linear code $C$ (right $R$-submodule of $R^n$), the orthogonal set of $C$ is defined by
$$C^\perp=\{v\in R^n~|~[v,c]=0~\forall~c\in C\}.$$ Note that if $C$ be a linear code (right $R$-submodule of $R^n$) then $C^\perp$ is a left $R$-submodule of $R^n$. It is well known that $|C||C^\perp|=|R^n|$ (see \cite{Woo99}).\\
For a right $R$-submodule $C$ of $R^n,$ the left annihilator of $C$ is deﬁned by $$\texttt{Ann}_{l}(C)=\{x\in R^n~|~[x, c]=0~\forall~c\in C\}.$$ Note that $\texttt{Ann}_{l}(C)$ is a left $R$-submodule of $R^n.$ The following proposition, we found in \cite{Woo99,Lam19}.
\begin{prop}\label{p-2.1}
 Let $C$ be a right $R$-submodule of $R^n$. Assume the left annihilator of $C$ is  $\texttt{Ann}_{l}(C).$ Then $|C||\texttt{Ann}_{l}(C)|=|R^n|.$
\end{prop}
Let $M$ be a right $R$-submodule of $R^n$ and $I$ be a $2$-sided ideal of $R.$ Recall that
$$MI=\{\sum_{\text{finite~sum}}mr~|~m\in M,~r\in I\}.$$
We give a proposition as follows.
\begin{prop}\label{p-5}\cite[Th. 13.11]{Is93}
 Let $M$ be a finitely generated right $R$-module. If $M\texttt{J}(R)=M,$ then $M=0.$
\end{prop}
We deduced the following proposition using Proposition \ref{p-5} and Theorem V.5 in \cite{McD74}.
\begin{prop}\label{p-6}
Let $R$ be a local$~^{*}$ ring. Let $t_1,t_2,\dots,t_k\in M,$ where $M$ is a right $R$-module.
\begin{itemize}
    \item[(a)] $t_1,t_2,\dots,t_k$ generates $M$ as a right $R$-module if and only if $\bar{t_1},\bar{t_2},\dots,\bar{t_k}$ generate $M/M\texttt{J}(R)$ as an $\mathbb{F}_q$-vector space.
    \item[(b)]  $t_1,t_2,\dots,t_k$ forms a minimal set of generators $M$ as a right $R$-module if and only if $\bar{t_1},\bar{t_2},\dots,\bar{t_k}$ generate $M/M\texttt{J}(R)$ as an $\mathbb{F}_q$-vector space.
\end{itemize}
\end{prop}
 $^{*}$ Reader can see Section $3$ about local rings.
\section{Characterization of LCP of codes over a finite non-commutative Frobenius local ring}

In this section, let $R$ be a finite non-commutative Frobenius local ring. According to Theorem $V.1$ in\cite{McD74}, $R$ is a local ring if and only if the non-units of $R$ form an additive Abelian group. We denote the Jacobson radical of $R$ is $\texttt{J}(R)$ consists of all non-unit elements of $R,$ as $R$ is an Artinian local ring. According to Wedderburn's little theorem, every skew field is commutative. Hence, we will denote $\mathbb{F}_q=R/\texttt{J}(R)$ as the residue field of $R.$ There is a natural surjective homomorphism from $R$ onto $\mathbb{F}_q,$ i.e., $\pi~:~R\rightarrow R/\texttt{J}(R),$ $r\mapsto r+\texttt{J}(R),$ for any $r$ in $R.$ This $\pi$ can be extended naturally to a homomorphism from $R^n$ to $\mathbb{F}_q^n.$ It is clear that this map is a surjective ring homomorphism. It is obvious that $\pi$ maps a linear code over $R$ to linear code over $\mathbb{F}_q.$ We deployed the preliminaries and existing results in Section $2.$ Let us note that linear code over $R$ is a right $R$-submodule of $R^n.$
On the other hand, since $R$ is a finite ring, then there exists $t\in \mathbb{N}$ such that $\texttt{J}(R)^{t+1}=\texttt{J}(R)^{t}.$ By the Proposition \ref{p-5}, we obtain that $\texttt{J}(R)^{t}=\{0\}.$ This implies that $\texttt{J}(R)$ is nilpotent with nilpotency index $t$ (i.e., $\texttt{J}(R)^{t}=\{0\}$ but $\texttt{J}(R)^{t-1}\neq\{0\}.$) Note that $\texttt{J}(R)$ is a both sided ideal of $R$ and satisfying the following condition
$$0=\texttt{J}(R)^{t}\subseteq \texttt{J}(R)^{t-1}\subseteq\dots\subseteq\texttt{J}(R)\subseteq R.$$
Using this relation, we make the following lemma.
\begin{lem}\label{lm-1}
Let $R$ be a finite non-commutative local ring with Jacobson radical $\texttt{J}(R).$ Then there exists $m\in \texttt{J}(R)$ such that $\alpha m=0,$ for all $\alpha\in \texttt{J}(R).$
\end{lem}
\begin{proof}
By the previous discussion, we obtain that $\texttt{J}(R)^t=0$ but $\texttt{J}(R)^{t-1}\neq 0$ and $\texttt{J}(R)^{t-1}\subseteq \texttt{J}(R).$ Thus there exists $m(\neq 0)\in \texttt{J}(R)^{t-1}$ such that $\alpha m=0,$ for all $\alpha\in \texttt{J}(R).$
\end{proof}
\begin{defi}\label{def-3.1}
 For any positive integer $n,$ an $n\times n$ matrix $\mathrm{A}$ is called invertible over $R$ if $\pi(\mathrm{A})$ is invertible over $\mathbb{F}_q,$ where $\pi(\mathrm{A})=(\pi(a_{ij}))$ for all $1\leq i,j \leq n.$
\end{defi}

\begin{lem}\label{lm-3.2}
 Let $C$ and $D$ be two linear codes over $R.$ Then $C\cap D=\{0\}$ if and only if $\pi(C)\cap \pi(D)=\{0\}.$
\end{lem}
\begin{proof}
Let us suppose that $\pi(C)\cap \pi(D)=\{0\}.$  We shall prove that $C\cap D=\{0\}.$  Let $x\in C\cap D,$ which implies $x\in C$ and $x\in D,$ this force that $\pi(x)\in \pi(C)\cap \pi(D).$ By the hypothesis, $\pi(x)=0,$ which means $x\in (C\cap D)\texttt{J}(R).$ Hence, $C\cap D =(C\cap D)\texttt{J}(R).$ Since $C\cap D$ is finitely generated right $R$-module, then by Proposition \ref{p-5}, we obtain that $C\cap D=\{0\}.$\\
 For the reverse part, let us assume that $C\cap D=\{0\}.$ Let $x\in \pi(C)\cap \pi(D),$ which means $x\in \pi (C)$ and $x\in\pi(D).$ then there exists $c\in C$ and $d\in D$ such that $x=\pi(c)=\pi(d).$ This follows that $c-d\in R^n\texttt{J}(R).$ Now by the Lemma \ref{lm-1}, there exist non-zero $m\in \texttt{J}(R)$ such that $(c-d)m=0,$ which implies $cm=dm\in C\cap D.$ By hypothesis $cm=dm=0.$ Thus, $c\in R^n\texttt{J}(R).$ Otherwise, $c\in R^n\setminus R^n\texttt{J}(R).$ This means that $cm (\neq 0) \in C\cap D,$ which contradicts the fact that $C\cap D=\{0\}.$ Therefore, $\pi(C)\cap \pi(D)=\{0\}.$
\end{proof}
\begin{thm}\label{th-1}
Let $C$ and $D$ be two linear codes over $R.$ If the pair $(C,D)$ is LCP. Then the pair $(\pi(C), \pi(D))$ is LCP.
\end{thm}
\begin{proof}
  Since the pair $(C,D)$ is an LCP code. i.e., $C+D=R^n$ and $C\cap D=\{0\}.$ By the Lemma \ref{lm-3.2}, we obtain that $\pi(C)\cap\pi(D)=\{0\}.$ Since, $\pi$ is surjective, for each $x\in \mathbb{F}_q^n,$ there exists $a\in R^n$ such that $x=\pi(a).$ Also, $a=c+d,$ for some $c\in C$ and $d\in D.$ Thus $x=\pi(c)+\pi(d)\in \pi(C)+\pi(D).$ Since $\pi(C)+\pi(D)$ is a subspace of $\mathbb{F}_q^n,$ then, $\pi(C)+\pi(D)=\mathbb{F}_q^n.$ Thus, the pair $(\pi(C), \pi(D))$ is an LCP.
\end{proof}
 Naturally one question arrive our mind for what condition LCP of codes
exist. Now, we will find answer of our question as follows. Before, we will state a lemma.
\begin{lem}\label{kap}\cite[Theorem 2]{Kap58}
Any projective module over a local ring (not necessarily commutative) is free right $R$-module.
\end{lem}
\begin{lem} Let $C$ and $D$ be two linear codes over $R.$ If the pair $(C,D)$ is an LCP of codes in $R^n,$ then $C$ and $D$ both are free right $R$-module.
\end{lem}
\begin{proof}
 Since the pair $(C,D)$ forms an LCP of codes, that means $C\oplus D=R^n$ it follows that $C\oplus D$ is free right $R$-module. This implies that $C$ and $D$ both are projective right $R$-module. As $R$ is local, then by Lemma \ref{kap}, $C$ and $D$ both are free right $R$-module.
\end{proof}
If we consider $C$ and $D$ to be both free right $R$-module, the pair $(C,D)$ may or may not be LCP of codes in general. For this, we illustrate an example as follows
\begin{exam}\label{ex-3.1}
Let $R$ be the collection of $2\times 2$ matrices, and the element is of the form as follows.

$$\left(
{\begin{array}{cccc}
   a & x \\
   0 & a \\
   \end{array} } \right),$$ where $a,x\in \mathbb{F}_q.$ The Jacobson radical of $R$ consists of matrices of the form $$\left(
{\begin{array}{cccc}
   0 & x \\
   0 & 0 \\
   \end{array} } \right),$$ where $x\in \mathbb{F}_q.$ One can check that $R/\texttt{J}(R)\simeq \mathbb{F}_q$ under the mapping $\left(
{\begin{array}{cccc}
   a & x \\
   0 & a \\
   \end{array} } \right)\mapsto a.$ Hence, $R$ is a non-commutative local ring. Let $C$ and $D$ be two linear codes over $R$ of length $4.$ The generator matrices of $C$ and $D$ are
   $G_1=\left(
{\begin{array}{cccc}
   \textbf{1} & \textbf{0} & \textbf{1} & \textbf{0} \\
   \textbf{0} & \textbf{1} & \textbf{0} & \textbf{1} \\
   \end{array} } \right)$ and $G_1=\left(
{\begin{array}{cccc}
   \textbf{1} & \textbf{1} & \textbf{1} & \textbf{1} \\
   \textbf{0} & \textbf{1} & \textbf{1} & \textbf{1} \\
   \end{array} } \right),$ where $\textbf{1}=\left(
{\begin{array}{cccc}
   1 & 0 \\
   0 & 1 \\
   \end{array} } \right)$ is unit in $R$. It is easy to see that $C+D=R^n$ but $C\cap D\neq \{0\}$ as $(\textbf{1},\textbf{1},\textbf{1},\textbf{1})\in C\cap D$. However, $C$ and $D$ are free.
   \end{exam}
\begin{thm}\label{th-2}
  Let $C$ and $D$ be two free linear codes over $R.$ Then the pair $(C,D)$ is LCP if and only if the pair $(\pi(C), \pi(D))$ is LCP.
\end{thm}
\begin{proof}
  From Theorem \ref{th-1}, follows that if the pair $(C,D)$ is LCP, then the pair $(\pi(C), \pi(D))$ is LCP.\\
  For the converse part, let us suppose that the pair $(\pi(C), \pi(D))$ is LCP, that means $\pi(C)\oplus\pi(D)=\mathbb{F}_q^n,$ i.e., $\pi(C)+\pi(D)=\mathbb{F}_q^n$ and $\pi(C)\cap \pi(D)=\{0\}.$ Now by the Lemma \ref{lm-3.2}, we obtain that $C\cap D=\{0\}.$ Let $\{\pi(x_1),\pi(x_2),\dots,\pi(x_k)\}$ be a basis of $\pi(C)$ and $\{\pi(x_{k+1}),\pi(x_{k+2}),\dots,\pi(x_n)\}$ be a basis of $\pi(D).$ By Proposition \ref{p-6}, we have $\{x_1,x_2,\dots,x_k\}$ is a minimal generating set of $C$ and $\{x_{k+1},x_{k+2},\dots,x_{n}\}$ is a minimal generating set of $D$. By hypothesis, $C$ and $D$ both are free, so we get that $|C||D|=|R^n|.$ Thus, the pair $(C,D)$ is LCP.
\end{proof}
Note that if $C$ is a linear code over $R,$ with generator matrix $\mathrm{G}$ and parity check matrix $\mathrm{H}$ then by Proposition \ref{p-6}, we have $\pi(\mathrm{G})$ and $\pi(\mathrm{H})$ is a generator and parity check matrix of $\pi(C),$ respectively. Now, we make a proposition, which is found in \cite{Bhowmick23}.
\begin{prop}\label{p-3.1}\cite[Corollary 4]{Bhowmick23}
 Let $C$ and $D$ are free codes over finite commutative Frobenius ring $R$ with generator matrices $\mathrm{G}_1,$ $\mathrm{G}_2$ and parity check matrices $\mathrm{H}_1,$ $\mathrm{H}_2,$ respectively, with the condition $|C||D|=|R^n|.$ Then the following statements are equivalent
 \begin{itemize}
     \item[(1)] the pair $(\pi(C),\pi(D))$ is an LCP codes in $\mathbb{F}_q^n$;
     \item[(2)] $\pi(\mathrm{H}_2)\pi(\mathrm{G}_1)^\top$ is invertible or
$\pi(\mathrm{H}_1)\pi(\mathrm{G}_2)^\top$ is invertible;
     \item[(3)] $\left(
{\begin{array}{ccc}
   \pi(\mathrm{G}_1) \\
   \pi(\mathrm{G}_2) \\
   \end{array} } \right)$ is invertible over $\mathbb{F}_q$;
   \item[(4)]  $\left(
{\begin{array}{ccc}
   \pi(\mathrm{H}_1) \\
   \pi(\mathrm{H}_2) \\
   \end{array} } \right)$ is invertible over $\mathbb{F}_q.$
 \end{itemize}

\end{prop}

\begin{thm}
 Let $C$ and $D$ be two free linear codes in $R^n$ with generator matrices $\mathrm{G}_1,$  $\mathrm{G}_2$ and parity check matrix $\mathrm{H}_1,$ $\mathrm{H}_2$ respectively, with the condition $|C||D|=|R^n|.$ Then the pair $(C,D)$ is LCP if and only if $\mathrm{H}_2\mathrm{G}_1^\top$ or $\mathrm{H}_1\mathrm{G}_2^\top$ are invertible.
\end{thm}
\begin{proof}
 Suppose the pair $(C,D)$ is an LCP of codes, i.e., $C\cap D=\{0\}$ and $C+D=R^n.$ If possible let $\mathrm{H}_2\mathrm{G}_1^\top$ is not invertible over $R$. From Definition \ref{def-3.1}, $\pi(\mathrm{H}_2\mathrm{G}_1^\top)$ is not invertible over $\mathbb{F}_q.$ Since, $\pi$ is a ring homomorphism, which implies $\pi(\mathrm{H}_2\mathrm{G}_1^\top)=\pi(\mathrm{H}_2) \pi(\mathrm{G}_1)^\top.$ Using the Proposition \ref{p-3.1}, we obtain that $(\pi(C),\pi(D))$ can not be an LCP of codes. Therefore, by the Theorem \ref{th-2}, we get that $(C, D)$ can not be an LCP of codes, which is a contradiction. Thus, $\mathrm{H}_2\mathrm{G}_1^\top$ or $\mathrm{H}_1\mathrm{G}_2^\top$ are invertible. \\
 Conversely, let us assume that $\mathrm{H}_2\mathrm{G}_1^\top$ or $\mathrm{H}_1\mathrm{G}_2^\top$ are invertible. By the Definition \ref{def-3.1}, $\pi(\mathrm{H}_2\mathrm{G}_1^\top)$ is invertible over $\mathbb{F}_q.$ Since, $\pi$ is a ring homomorphism, which implies $\pi(\mathrm{H}_2\mathrm{G}_1^\top)=\pi(\mathrm{H}_2) \pi(\mathrm{G}_1)^\top.$ This implies that $(\pi(C), \pi(D))$ is LCP by using Proposition \ref{p-3.1}. By the Theorem \ref{th-2}, we obtain that the pair $(C, D)$ forms an LCP of codes over $R.$
 \end{proof}

\begin{thm}
 Let $C$ and $D$ be two free linear codes in $R^n$ with generator matrices $\mathrm{G}_1,$  $\mathrm{G}_2$ and parity check matrix $\mathrm{H}_1,$ $\mathrm{H}_2$ respectively, with the condition $|C||D|=|R^n|.$ Then the
following are equivalent
\begin{itemize}
\item[(1)] the pair $(C, D)$ form LCP; \item[(2)] $\left(
{\begin{array}{ccc}
   \mathrm{G}_1 \\
   \mathrm{G}_2 \\
   \end{array} } \right)$ is invertible over $R$;
\item[(3)] $\left( {\begin{array}{ccc}
   \mathrm{H}_1 \\
   \mathrm{H}_2 \\
   \end{array} } \right)$ is invertible over $R$.
\end{itemize}
\end{thm}
\begin{proof}
  We shall prove this theorem in the following sequence $(1)$ implies $(2),$  $(2)$ implies $(1)$ and $(1)$ implies $(3),$  $(3)$ implies $(1).$ We only need to show that $(1)$ implies $(2)$  and vice-versa. The reader can prove another consequence similarly.\\
  For the proof, $(1)$ implies $(2),$ let us suppose that the pair $(C, D)$ forms LCP. If possible, let us suppose that  $\left(
{\begin{array}{ccc}
   \mathrm{G}_1 \\
   \mathrm{G}_2 \\
   \end{array} } \right)$ is not invertible over $R.$  From Definition \ref{def-3.1}, $\pi\left(
{\begin{array}{ccc}
   \mathrm{G}_1 \\
   \mathrm{G}_2 \\
   \end{array} } \right)$ is not invertible over $\mathbb{F}_q.$ Since, $\pi$ is a ring homomorphism, which implies $\pi\left(
{\begin{array}{ccc}
   \mathrm{G}_1 \\
   \mathrm{G}_2 \\
   \end{array} } \right)=\left(
{\begin{array}{ccc}
   \pi(\mathrm{G}_1) \\
   \pi(\mathrm{G}_2) \\
   \end{array} } \right).$ Then the matrix $\left(
{\begin{array}{ccc}
   \pi(\mathrm{G}_1) \\
   \pi(\mathrm{G}_2) \\
   \end{array} } \right)$ is not invertible over $\mathbb{F}_q.$ Hence the pair $(\pi(C),\pi(D))$ can not be an LCD, follows from Proposition \ref{p-3.1}. This  contradict to the fact the Theorem \ref{th-2}.\\
  For the proof of $(2)$ implies $(1),$ let us suppose that $\left(
{\begin{array}{ccc}
   \mathrm{G}_1 \\
   \mathrm{G}_2 \\
   \end{array} } \right)$ is invertible over $R.$ From Definition \ref{def-3.1}, $\pi\left(
{\begin{array}{ccc}
   \mathrm{G}_1 \\
   \mathrm{G}_2 \\
   \end{array} } \right)$ is invertible over $\mathbb{F}_q.$ Since, $\pi$ is a ring homomorphism, which implies $\pi\left(
{\begin{array}{ccc}
   \mathrm{G}_1 \\
   \mathrm{G}_2 \\
   \end{array} } \right)=\left(
{\begin{array}{ccc}
   \pi(\mathrm{G}_1) \\
   \pi(\mathrm{G}_2) \\
   \end{array} } \right).$ Then the matrix $\left(
{\begin{array}{ccc}
   \pi(\mathrm{G}_1) \\
   \pi(\mathrm{G}_2) \\
   \end{array} } \right)$ is invertible over $\mathbb{F}_q.$ Hence the pair $(\pi(C),\pi(D))$ forms an LCD, follows from Proposition \ref{p-3.1}. Therefore, by the Theorem \ref{th-2}, the pair $(C,D)$ is an LCP.
\end{proof}

\section{LCP codes over a finite non-commutative Frobenius ring}
In this section, we used the symbol $R$ as a finite non-commutative Frobenius ring. For the codes $C$ and $D$, the pair $(C,D)$ is called LCP if $C\oplus D=R^n.$ This gives that $C$ and $D$ both are projective right $R$-module. Towards this, $C$ and $D$ are injective right $R$-module, it follows from Proposition \ref{th-0.02}. Consequently, if a pair $(C,D)$ is LCP, then $C$ and $D$ both are injective right $R$-module. If $C$ and $D$ both are injective right $R$-module, the pair $(C,D)$ may or may not be LCP. Readers can find an easy example. But, we claim that if injective hull of $C\oplus D$ is $R^n,$ then the pair $(C,D)$ forms LCP. We will answer this question later. Now we will introduce an example of LCP codes $(C,D)$ over $R.$ Note that both the code $C$ and $D$ are non-free right $R$-submodule of $R^n.$
\begin{exam}\label{ex-4.1}
 Let $\mathbb{F}_q$ be a finite field with $q$ elements. Let $R$ be the $4$-dimensional $\mathbb{F}_q$-ring consisting of matrices  of the form $$\left(
{\begin{array}{cccc}
   a & x & 0 & 0\\
   0 & a & 0 & 0\\
   0 & 0 & b & y\\
   0 & 0 & 0 & b\\
   \end{array} } \right),$$ where $a,b,x,y\in \mathbb{F}_q.$ The Jacobson radical of $R$ is the collection of matrices of the form
   $$\left(
{\begin{array}{cccc}
   0 & x & 0 & 0\\
   0 & 0 & 0 & 0\\
   0 & 0 & 0 & y\\
   0 & 0 & 0 & 0\\
   \end{array} } \right).$$ Reader can check that $R/\texttt{J}(R)\simeq \mathbb{F}_q \times \mathbb{F}_q,$ under the mapping $\left(
{\begin{array}{cccc}
   a & x & 0 & 0\\
   0 & a & 0 & 0\\
   0 & 0 & b & y\\
   0 & 0 & 0 & b\\
   \end{array} } \right)\mapsto (a,b).$ Therefore, $R$ is a finite non-commutative Frobenius ring. Let us consider $e_1, e_2\in R$ such that $e_1R$ and $e_2R$ forms a linear codes over $R,$ where
   $e_1=\left(
{\begin{array}{cccc}
   1 & 0 & 0 & 0\\
   0 & 1 & 0 & 0\\
   0 & 0 & 0 & 0\\
   0 & 0 & 0 & 0\\
   \end{array} } \right)$ and $e_2=\left(
{\begin{array}{cccc}
   0 & 0 & 0 & 0\\
   0 & 0 & 0 & 0\\
   0 & 0 & 1 & 0\\
   0 & 0 & 0 & 1\\
   \end{array} } \right).$ Let us choose $C=e_1R$ and $D=e_2R.$ Since $e_1$ and $e_2$ both are idempotent and $e_1+e_2=\textbf{1},$ where $\textbf{1}=\left(
{\begin{array}{cccc}
   1 & 0 & 0 & 0\\
   0 & 1 & 0 & 0\\
   0 & 0 & 1 & 0\\
   0 & 0 & 0 & 1\\
   \end{array} } \right)$ is identity in $R.$ Thus, $C\oplus D=R.$ Hence the pair $(C, D)$ forms an LCP. But neither $C$ nor $D$ are free. Otherwise, $C=R,$ which is absurd.
\end{exam}
In Example \ref{ex-4.1}, we have seen that non-free LCP codes exist over finite Frobenius rings. Now, we will characterize LCP codes over finite non-commutative Frobenius rings using advanced module theory.
\begin{defi}
  Let $C$ be a right $R$-submodule of $R^n.$ $C$ is said to be an essential submodule of $R^n$ if for every non-zero right $R$-submodule $M$ of $R^n$ such that $C\cap M\neq 0.$
\end{defi}
\begin{lem}\label{lm-4.1}
 Let $C$ and $D$ be two linear codes in $R^n.$ If the pair $(C,D)$ is an LCP, then $C\oplus D$ is essential submodule of $R^n.$
\end{lem}
\begin{proof}
Since, the pair $(C,D)$ forms an LCP, i.e., $C\oplus D=R^n.$ It means if there exists a right $R$-sbmodule $M$ of $R^n$ such that $(C\oplus D)\cap M=\{0\},$ which implies $M=0.$ Thus, $C\oplus D$ is essential submodule of $R^n.$
\end{proof}
For Reader convenience, we will find converse of the above lemme later. Before, we will recall some preliminary that have an important role in finding the answer.
\begin{defi}
 Let $C$ be a right $R$-submodule of $R^n.$ A right $R$-submodule $I$ is said to be an injective hull of $C$ if $I$ is injective module and $C$ is an essential submodule of $I.$ We denote injective hull of $C$ as $E(C)=I$
\end{defi}
\begin{thm}
 Let $C$ and $D$ be two linear codes in $R^n.$ Then the pair $(C,D)$ is LCP if and only if $C$ and $D$ both are injective with $E(C\oplus D)=R^n.$
\end{thm}
\begin{proof}
  Let us suppose that the pair $(C,D)$ is an LCP, i.e., $C\oplus D=R^n.$ Therefore, $C$ and $D$ are injective. By the Lemma \ref{lm-4.1}, $C\oplus D$ is essential submodule of $R^n,$ thus $E(C\oplus D)=R^n.$ \\
  Conversely, let us suppose that $C$ and $D$ both are injective with $E(C\oplus D)=R^n.$ i.e., $E(C)=C$ and $E(D)=D$ also $C\oplus D=E(C)\oplus E(D)=E(C\oplus D)=R^n.$ Thus, $(C, D)$ forms LCP.
\end{proof}
\begin{defi}
Let $C$ and $D$ be two right $R$-submodule of $R^n.$ $D$ is said to be a complement submodule of $C$ in $R^n$ if $D$ is a right $R$-submodule of $R^n$ maximal with respect to $C\cap D=\{0\}.$
\end{defi}
By the routine application of Zorn's lemma, we see that for a given right $R$-submodule $C$ of $R^n$ there always exists a maximal submodule $D$ of $R^n$ such that $C\cap D=\{0\}.$ Consequently, $C\oplus D$ is a right $R$-submodule of $R^n.$
\begin{lem}\label{lm-4.2}
 Let $C$ and $D$ be two linear codes in $R^n.$ If the pair $(C,D)$ is an LCP, then $C$ is a complement submodule of $D$ in $R^n$ and $D$ is a complement submodule of $C$ in $R^n.$
\end{lem}
\begin{proof}
Since, the pair $(C,D)$ forms an LCP, i.e., $C\oplus D=R^n.$ It means $C+D=R^n$ and $C\cap D=\{0\}.$ This shows that $C$ is a complement submodule of $D$ in $R^n$ and $D$ is a complement submodule of $C$ in $R^n.$
\end{proof}
Note that for a right $R$-submodule $C$ of $R^n$ is essential submodule of $R^n$ if and only if $0$ is a maximal submodule of $R^n.$
\begin{defi}\label{de-4.2}
 Let $C$ be a right $R$-submodule of $R^n.$ $C$ is said to be a closed in $R^n$ if $C$ is the complement of some submodule of $R^n.$
\end{defi}
A  right $R$-submodule of $R^n$ is closed, provided $C$ has no proper essential submodule in $R^n.$ Moreover, if $D$ is any right $R$-submodule of $R^n,$ then there exists, by Zorn's Lemma, a submodule $K$ of $R^n$ maximal with respect to the property that $D$ is an essential submodule of $C$ and in this case, $C$ is a closed submodule of $R^n.$
\begin{lem}\label{lm-10}
 Let $C$ be a right $R$-submodule of $R^n.$ If $C$ is essential submodule of $R^n$ and closed submodule of $R^n.$ Then $C=R^n.$
\end{lem}
\begin{proof}
 Since $C$ is an essential submodule of $R^n,$ then for every non-zero right $R$-submodule of $R^n$ such that $C\cap M\neq\{0\}.$ Again, $C$ is closed then there exist some right $R$-submodule $N$ such that $N$ is a complement of $C.$ This force that $C=R^n.$
\end{proof}
\begin{lem}
Let $C$ and $D$ be two linear codes in $R^n.$ If the pair $(C,D)$ is an LCP, then $C$ and $D$ both are closed submodule of $R^n.$
\end{lem}
\begin{proof}
Since, the pair $(C,D)$ forms an LCP, i.e., $C\oplus D=R^n.$ It means $|C||D|=|R^n|$ and $C\cap D=\{0\}.$ Thus, $C$ is a complement of $D$ and $D$ is also complement of $C.$ This shows that $C$ and $D$ both are closed submodule of $R^n.$
\end{proof}
\begin{thm}
Let $C$ and $D$ be two linear codes in $R^n.$ Then the pair $(C,D)$ is LCP if and only if $D$ is complement of $C$ such that $C\oplus D$ is closed.
\end{thm}
\begin{proof}
Let us suppose that the pair $(C,D)$ is LCP. Therefore, $C\cap D=R^n.$ and $C\cap D=\{0\}.$ Then by the Lemma \ref{lm-4.2}, $D$ is complement of $C$. Since $C\oplus D=R^n,$ it follows that $C\oplus D$ is a complement of $\{0\}.$ From Definition \ref{de-4.2}, we get $C\oplus D$ is closed.\\
Conversely, let us assume that $D$ is a complement of $C$ such that $C\oplus D$ is closed. We shall show that $C\oplus D$ is an essential submodule of $R^n.$ Let $M$ be a non-zero right $R$-submodule of $R^n.$ We have to see that $(C\oplus D)\cap M\neq\{0\}.$ This is clear if $M\subseteq D.$ Otherwise, the maximality condition of $D$ that $C\cap (D+M)\neq \{0\},$ which implies there exists a non-zero $x\in C\cap (D+M),$ this gives that $x=d+m,$ where $d\in D$ and $m\in M.$ Hence $m(\neq 0)\in (C\oplus D)\cap M,$ as $C\cap D=\{0\}.$ By the hypothesis, $C\oplus D$ is a closed. From the Lemma \ref{lm-10}, we have $C\oplus D=R^n.$ This gives that the pair $(C, D)$ forms an LCP codes over $R.$
\end{proof}
\begin{thm}
 Let $C$ and $D$ be two linear codes in $R^n$ such that $C\oplus D$ is closed. Then the following are equivalent
 \begin{itemize}
 \item[(1)] the pair $(C, D)$ is LCP;
 \item[(2)] $C\oplus D$ is an essential submodule of $R^n$;
 \item[(3)] $0$ is a complement of $C\oplus D.$
 \end{itemize}
\end{thm}
\begin{proof}
 $(1)\Rightarrow(2)$ It follows from Lemma \ref{lm-4.1}.\\
 $(2)\Rightarrow(1)$ By the hypothesis, $C\oplus D$ is closed and $C\oplus D$ is an essential submodule of $R^n,$ thus, by Lemma \ref{lm-10}, $(C,D)$ forms an LCD code.\\
 $(2)\Leftrightarrow(3)$ It is trivial.
\end{proof}
\begin{thm}\label{th-8}
 Let $C$ and $D$ be two linear codes in $R^n$ such that $|C||D|=|R^n|.$ Then the following are equivalent
 \begin{itemize}
 \item[(1)] the pair $(C, D)$ is LCP;
 \item[(2)] $C\oplus D$ is an essential submodule of $R^n$;
 \item[(3)] $0$ is a complement of $C\oplus D$;
 \item[(4)] $D$ is complement of $C$;
 \item[(5)] $C\oplus D$ is closed.
 \end{itemize}
\end{thm}
\begin{proof}
 $(1)\Rightarrow(2)$ It follows from Lemma \ref{lm-4.1}.\\
 $(2)\Rightarrow(1)$ By the hypothesis, $|C||D|=|R^n|$ and $C\oplus D$ is a submodule of $R^n,$ thus $(C,D)$ forms an LCD code.\\
 $(2)\Leftrightarrow(3)$ It is trivial.\\
 $(1)\Rightarrow(4)$ It is follows from Lemma \ref{lm-4.2}.\\
 $(4)\Rightarrow(2)$ Let $N$ be a non-zero right $R$-submodule of $R^n.$ We have to see that $(C\oplus D)\cap N\neq\{0\}.$ This is clear if $N\subseteq D.$ Otherwise, the maximality condition of $D$ that $C\cap (D+N)\neq \{0\},$ which implies there exists a non-zero $x\in C\cap (D+N),$ this gives that $x=d+y,$ where $d\in D$ and $y\in N.$ Hence $y(\neq 0)\in (C\oplus D)\cap N,$ as $C\cap D=\{0\}.$\\
 $(1)\Rightarrow(5)$ Since $(C, D)$ is LCP, i.e., $C\oplus D=R^n.$ It gives that $C\oplus D$ is a complement of $\{0\}$ in $R^n.$ By Definition \ref{de-4.2}, $C\oplus D$ is close submodule of $R^n.$\\
 $(5)\Rightarrow(1)$ It is trivial.
\end{proof}
\begin{rem}
Let $C$ and $D$ be two linear codes over a finite non-commutative ring $R.$ If the pair $(C, D)$ satisfies the Theorem \ref{th-8}, then $(C,D)$ forms a non-free LCP codes. For example:
\end{rem}
\begin{exam}
  In Example \ref{ex-4.1}, $C=e_1R$ and $D=e_2R$ both are non-free linear codes over $R.$ One can check that the pair $(C,D)$ satisfies the all condition of Theorem \ref{th-8}. Hence $(C,D)$ is an example of non-free LCP codes over $R.$
\end{exam}

\section{Equivalent codes}
In this section, our main aim is to find conditions for a pair of LCP codes $(C, D)$ such that $D^\perp$ is equivalent to $C.$ We denote $R$ as a finite non-commutative Frobenius ring. Since $R$ is finite, any right $R$-submodule of $R^n$ is finitely generated. Let us suppose that $\{\alpha_1,\alpha_2,\dots,\alpha_k\}$ be a minimal generating set of $C.$ Let $\mathrm{G}$ be a generator matrix of $C$ and $G$ is of the form $$\mathrm{G}=\left(
{\begin{array}{cccccc}
   \alpha_1\\
   \alpha_2\\
   \vdots\\
  \alpha_k\\
   \end{array} } \right)=\left(
{\begin{array}{cccc}
   \alpha_{11} & \alpha_{12} & \cdots & \alpha_{1n}\\
   \alpha_{21} & \alpha_{22} & \cdots & \alpha_{2n}\\
   \vdots & \vdots & \ddots &\vdots\\
   \alpha_{k1} & \alpha_{k2} & \cdots & \alpha_{kn}\\
   \end{array} } \right).$$ Thus, $G$ is an $k\times n$ matrix, where $k\leq n.$ Any linear code $C$ is a right $R$-submodule of $R^n,$ with generator matrix $\mathrm{G}.$ Then $C$ can be written as follows
   $$C=\{\mathrm{G}^\top\alpha^\top~|~\alpha\in R^k\}.$$ Now we consider $e$ such that
   $$e=\left(
{\begin{array}{cccc}
   \alpha_{11} & \alpha_{12} & \cdots & \alpha_{1n}\\
   \alpha_{21} & \alpha_{22} & \cdots & \alpha_{2n}\\
   \vdots & \vdots & \ddots &\vdots\\
   \alpha_{k1} & \alpha_{k2} & \cdots & \alpha_{kn}\\
    0 & 0 & \cdots & 0\\
    \vdots & \vdots & \ddots &\vdots\\
    0 & 0 & \cdots & 0\\
    \end{array} } \right).$$ Now, $C$ can be rewritten as follows
    $$C=\{\mathrm{G}^\top\alpha^\top~|~\alpha\in R^k\}=\{e^\top\alpha^\top~|~\alpha\in R^n\}.$$
 \begin{lem}\label{lm-5.1}
  If $C=\{e^\top\alpha^\top~|~\alpha\in R^n\}$ is a linear code over $R,$ then the dual code of $C$ is $$C^\perp=\texttt{Ann}_{l}(C).$$  Moreover, $C^\perp$ is a left $R$-submodule of $R^n.$
 \end{lem}
 \begin{proof}
  $C$ is a right $R$-submodule of $R^n.$ Now, definition of $\texttt{Ann}_{l}(C)=\{x\in R^n~|~[x,c]=0~\forall ~c\in C\},$ follows that $\texttt{Ann}_{l}(C)\subseteq C^\perp.$ From the Proposition \ref{p-2.1}, we get
  $$|C^\perp|=\dfrac{|R^n|}{|C|}=|\texttt{Ann}_{l}(C)|.$$ Hence, $C^\perp=\texttt{Ann}_{l}(C).$
 \end{proof}
 \begin{lem}\label{lm-5.2} If $e^2=e\in M_n(R)$ and $C=\{e^\top\alpha^\top~|~\alpha\in R^n\},$ then $C^\perp=\{\beta(\textbf{1}-e^\top)~|~\beta\in R^n\}.$
 \end{lem}
 \begin{proof}
   By the Lemma \ref{lm-5.1}, we have $C^\perp=\{\beta(\textbf{1}-e^\top)~|~\beta\in R^n\},$ as $e^2=e\in M_n(R).$
 \end{proof}
\begin{thm} Let $C$ and $D$ be two linear codes over $R$ such that the pair $(C,D)$ forms an LCP. If there exists an idempotent element $e\in M_n(R)$ such that $C=\{e^\top\alpha^\top~|~\alpha\in R^n\}$ and $D$ is equivalent to $\{(\textbf{1}-e)\beta^\top~|~\beta\in R^n\}.$ Then $C^\perp$ is equivalent to $D.$ Moreover, $D^\perp$ is equivalent to $C.$
\end{thm}\label{th-9}
 \begin{proof}
  Since $e$ is idempotent element in $M_n(R),$ and $C=\{(e^\top\beta^\top~|~\beta\in R^n\},$ then by Lemma \ref{lm-5.2}, $C^\perp=\{\delta(\textbf{1}- e^\top)~|~\delta\in R^n\}.$ Since $(\alpha e)^\top=e^\top \alpha^\top.$ Thus, $C^\perp$ is equivalent to $D.$
\end{proof}
\begin{rem}
  In Theorem \ref{th-9}, we see that if $C$ and $D$ be two linear code over finite non-commutative Frobenius ring $R$ such that $C=\{e^\top\alpha^\top~|~\alpha\in R^n\}$ and $D=\{(\textbf{1}-e)\beta^\top~|~\beta\in R^n\},$ where $e^2=e.$ Then $C^\perp$ is equivalent to $D.$ If we replace a finite non-commutative Frobenius ring with a finite commutative Frobenius ring. The proof is similar. Moreover, Theorem \ref{th-9} is applicable for $R$ is a finite field and also for a finite chain ring or finite local ring.
\end{rem}
\begin{exam}
 Let $C$ and $D$ be two linear code over $R$ (Define in Example \ref{ex-3.1}) and $q=3$ of length $3$ with generator matrices $$\mathrm{G}=\left(
{\begin{array}{cccc}
   \textbf{1} & \textbf{2} & \textbf{0} \\
   \textbf{0} & \textbf{1} & \textbf{2} \\
   \end{array} } \right),~~~~~\mathrm{G}_2=\left(
{\begin{array}{cccc}
   \textbf{1} & \textbf{2} & \textbf{1} \\
   \end{array} } \right),$$  respectively, where $\textbf{1}=\left({\begin{array}{cccc}
   1 &  0 \\
   0 & 1 \\
   \end{array} } \right)\in R$ and $\textbf{2}=\left({\begin{array}{cccc}
   2 &  0 \\
   0 & 2 \\
   \end{array} } \right)\in R$. It is easy to see that $C\oplus D=R^n.$ Let us consider $$e=\left(
{\begin{array}{cccc}
   \textbf{1} &  \textbf{0} & \textbf{2}\\
   \textbf{0} & \textbf{1} & \textbf{2}\\
   \textbf{0} & \textbf{0} & \textbf{0}\\
   \end{array} } \right)$$ such that $e^2=e\in M_3(R)$ and $C=\{e^\top\alpha^\top~|~\alpha\in R^n\}$ and $D$ is equivalent to $\{(1-e)\beta^\top~|~\beta\in R^n\}.$ By the Theorem \ref{th-9}, we obtain $C^\perp$ is equivalent to $D.$ Moreover, $D^\perp$ is equivalent to $C.$
\end{exam}
\begin{exam}
    Let $C$ and $D$ be two linear code over $R$ (Define in Example \ref{ex-4.1}) and $q=3$ of length $4$ with generator matrices $$\mathrm{G}_1=\left(
{\begin{array}{cccc}
   \textbf{1} & \textbf{0} & \textbf{1} & \textbf{1} \\
   \textbf{0} & \textbf{1} & \textbf{1} & \textbf{1} \\
   \end{array} } \right),~~~~~\mathrm{G}_2=\left(
{\begin{array}{cccc}
   \textbf{1} & \textbf{1} & \textbf{1} & \textbf{0}\\
   \textbf{1} & \textbf{1} & \textbf{0} & \textbf{1}\\
   \end{array} } \right),$$  respectively. It is easy to see that $C\oplus D=R^n.$ Let us consider $$e=\left(
{\begin{array}{cccc}
   \textbf{1} &  \textbf{0} & \textbf{1} & \textbf{1}\\
   \textbf{0} & \textbf{1} & \textbf{1} &  \textbf{1}\\
   \textbf{0} & \textbf{0} & \textbf{0} & \textbf{0}\\
   \textbf{0} & \textbf{0} & \textbf{0} & \textbf{0}\\
   \end{array} } \right)$$ such that $e^2=e\in M_4(R)$ and $C=\{e^\top\alpha^\top~|~\alpha\in R^n\}$ and $D$ is equivalent to $\{(\textbf{1}-e)\beta^\top~|~\beta\in R^n\}.$ By the Theorem \ref{th-9}, we obtain $C^\perp$ is equivalent to $D.$ This gives $d(C^\perp)=d(D)$ Moreover, $D^\perp$ is equivalent to $C.$ i.e., $d(D^\perp)=d(C).$
\end{exam}
\section*{Acknowledgements} The first author would like to thank the National Institute of Science Education and Research Institute for its hospitality and support. The authors sincerely thank the referees for their careful reading and suggestions for improving our manuscript.

\section*{Declaration of Interest Statement} The authors declare that they have no known competing financial interests or personal relationships that could have appeared to influence the work reported in this article.

\end{document}